\documentclass[reqno,12pt,letterpaper]{amsart}
\usepackage{amsmath,amssymb,amsthm,graphicx,mathrsfs,url}
\usepackage[usenames,dvipsnames]{color}
\usepackage[colorlinks=true,linkcolor=Red,citecolor=Green]{hyperref}
\usepackage{amsxtra}

\def\?[#1]{\textbf{[#1]}\marginpar{\Large{\textbf{??}}}}

\let\epsilon=\varepsilon 

\newcommand{\RR}{{\mathbb R}}

\newcommand{\NN}{{\mathbb N}}
\newcommand{\CC}{{\mathbb C}}
\newcommand{\CI}{{{\mathcal C}^\infty}}
\newcommand{\CIc}{{{\mathcal C}^\infty_{\rm{c}}}}

\newtheorem{thm}{Theorem}
\newtheorem{prop}{Proposition}

\newtheorem{lem}[prop]{Lemma}

\numberwithin{equation}{section}

\let\Im=\Imag

\DeclareMathOperator{\rank}{rank}

\let\Re=\Real

\DeclareMathOperator{\supp}{supp}

\DeclareMathOperator{\tr}{tr}

\title{Scattering resonances as viscosity limits}
\author{Maciej Zworski}
\email{zworski@math.berkeley.edu}
\address{Department of Mathematics, University of California,
Berkeley, CA 94720, USA}

\begin{document}

\maketitle

\begin{abstract}
Using the method of complex scaling we show that scattering 
resonances of $ - \Delta + V $, $ V \in L^\infty_{\rm{c}} ( \RR^n ) $, 
are limits of eigenvalues of $ - \Delta + V - i \epsilon x^2 $ as
$ \epsilon \to 0+ $. That justifies a method proposed in 
computational chemistry and reflects a general principle for 
resonances in other settings.
\end{abstract}

\section{Introduction and statement of results}

In this note we show that scattering resonances can be defined as viscosity limits,
that is limits of eigenvalues of Hamiltonians suitably regularized as
infinity. The detailed proofs are presented in the simplest case
of the Schr\"odinger operator with a compactly supported potential
and rely only on standard techniques.

We consider 
\[  P := - \Delta + V , \ \ V \in L^\infty_{\rm{comp}} ( \RR^n ) , \]
where $  L^\infty_{\rm{comp}} $ denotes the spaces of bounded functions
vanishing outside of some compact set. (Similarly the subscript 
$ L^\bullet_{\rm{loc}} $ denotes function in the space $ L^\bullet $ on 
compact sets.) The scattering resonances
are defined as the poles of the meromorphic continuation of
resolvent:
\[ R_V ( z ) := ( - \Delta + V - z)^{-1} : L^2 ( \RR^n)  \to L^2 ( \RR^n ) , \ \ 
\Im z > 0 , \]
from the upper half-plane through the continuous spectrum. More precisely,
\begin{equation}
\label{eq:Rofz} 
 R_V ( z ) : L^2_{\rm{comp} } ( \RR^n ) \to L^2_{\rm{loc} } ( \RR^n ) , 
 \end{equation}
continues meromorphically to the double cover of $ \CC $ when $ n $ is 
odd and to the logarithmic cover of $ \CC $ when $ n $ is even. 
The poles of this continuation coincide with the poles of the scattering
matrix for the potential $ V$. Their multiplicity (except at the 
threshold $ z = 0 $) are given by 
\begin{equation}
\label{eq:mz}  m ( z ) := \rank \oint_z R_V ( \zeta ) d \zeta , 
\end{equation}
where the integration is over a small circle around $ z $ -- see \cite[Chapter 3]{res}.

Equivalently, we can consider Green's function, that is the integral
kernel of $ R_V ( z ) $, 
\begin{equation}
\label{eq:Green} R_V  ( z ) f  ( x ) = \int_{ \RR^n } G (z , x , y ) f ( y ) dy , 
\end{equation}
and look at the poles of the continuation of $ z \mapsto G ( z , x , y ) $
for $ x $ and $ y $ fixed. Another way, based on the method of complex scaling, 
will be reviewed in \S \ref{cs}. 

We now consider a {\em regularized} operator, 
\begin{equation}
\label{eq:Peps} 
 P_\epsilon := - \Delta + V - i \epsilon x^2 , \ \ \epsilon > 0 . 
\end{equation}
It is easy to see (with details reviewed in \S \ref{mc}) that $ P_\epsilon $
is an unbounded operator on $ L^2 ( \RR^n ) $ with a discrete spectrum. 
We have

\begin{thm} \label{t:1} Suppose that $ \{ z_j ( \epsilon) \}_{j=1}^\infty $ are 
the eigenvalues of $ P_\epsilon $. Then, uniformly on  
compact subsets of $ \{ z : -\pi/4 < \arg z < 7 \pi/4 \} $, 
\[   z_j ( \epsilon ) \to z_j ,  \ \ \epsilon \to 0 + , \]
where $ z_j $ are the resonances of $ P $. 
\end{thm}

\begin{figure}
\includegraphics{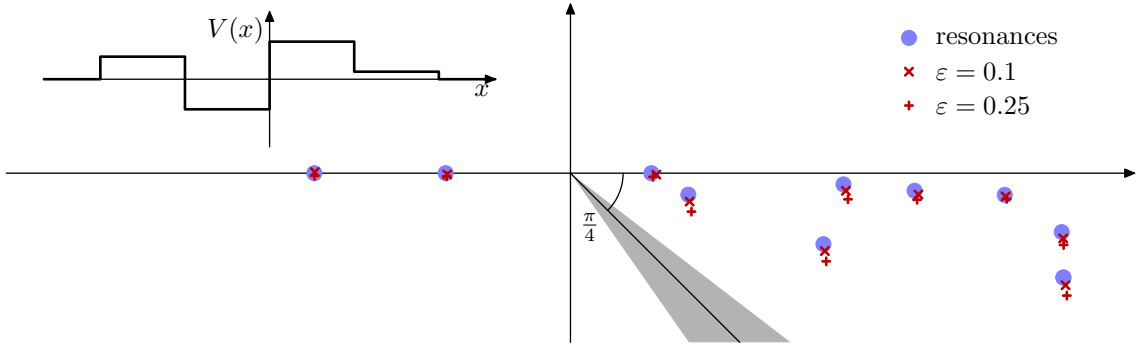} 
\caption{An illustration of the results of Theorem \ref{t:1} in the 
case of a specific potential shown on the left. Resonances
are computed using {\tt squarepot.m} \cite{BZ}. The eigenvalues
of $ P_\epsilon $, $ \epsilon = 1/4 $ and $ \epsilon = 1/10 $ 
are computing by approximating the unitarily equivalent operator
$ \epsilon^{\frac12} ( D_x^2 + \epsilon^{-\frac12} V ( \epsilon^{-\frac12} x ) 
- i x^2 ) $
using
a basis of eigenfunction of the harmonic oscillator $ D_x^2 + x^2 $.
}
\label{f:0}
\end{figure}

\noindent
{\bf Remarks.} 1. A more precise statement involving continuity of 
spectral projections is given in \S \ref{poc}. The term viscosity 
is motivated by the viscosity definition of Pollicott--Ruelle 
resonances given in Dyatlov--Zworski \cite{dam} -- see Example 3 below.

\noindent
2. When $ \epsilon < 0 $ the spectrum of $ P_\epsilon $ is given by 
complex conjugates of the spectrum of $ P_{- \epsilon } $. Hence
we have 
\begin{equation}
\label{eq:zje}
 z_j ( \epsilon ) \to \bar z_j , \ \  \epsilon \to 0 - , \end{equation}
 uniformly 
on compact subsets of $ \{ z : - 7 \pi/4 < \arg z < \pi/4 \}$. 

\noindent
3. The term $ - i \epsilon x^2 $ is an example of
a {\em complex absorbing potential} and other potentials can 
also be used -- see the discussion below. The proof here requires some analyticity properties of the complex absorbing potential.

\noindent
4. The restriction to $ \arg z > - \pi /4 $ when using $ - i \epsilon x^2 $
is due to the fact that for $ V \equiv 0 $ the spectrum of 
$ - \Delta - i \epsilon x^2 $ is given by $ \epsilon^{\frac12} e^{- \pi i/4} 
( 2 | \ell | + n ) $, $ \ell \in \NN^n $ which is a rescaled 
spectrum of the Davies harmonic oscillator -- see \S \ref{dho}. 
One can expand the range using $ \epsilon e^{ - i \alpha } x^2 $, $ 0 < \alpha < \pi $
in which case we recover resonances with $ \arg z > -\alpha/2  $. 

\noindent
5. The proof applies with simple modifications to compactly supported
{\em black box} perturbations on $ \RR^n $ introduced in \cite{SZ1} -- 
see \cite[Chapter 4]{res} and \cite{SjR}. In that case we need to 
replace $ - i \epsilon x^2 $ by $ - i \epsilon ( 1 - \chi ( x ) ) x^2 $ 
where $ \chi \in \CIc ( \RR^n ) $ is equal to $ 1 $ on a sufficiently 
large set -- see Example 2 below.

The computational method based on calculating eigenvalues of $ P_\epsilon $
was introduced in physical chemistry -- see 
Riss--Meyer \cite{RiMe} and Seideman--Miller \cite{semi} for the  
original approach and Jagau et al \cite{Jag} for some recent developments 
and references. 
However no rigorous mathematical treatments seem to be available and 
some  new interesting open questions can be posed -- see Example 4 below. 

Fixed complex absorbing potentials have already been used in mathematical literature
on scattering resonances. Stefanov \cite{SteCAP} showed that semiclassical 
resonances close to the real axis can be well approximated using 
eigenvalues of the Hamiltonian modified by a complex absorbing potential.
Nonnenmacher--Zworski \cite{NZ1},\cite{NZ2} used 
fixed complex absorbing potentials 
to study resonance problems employing gluing techniques of Datchev--Vasy
\cite{DV1},\cite{DV2}. Yet another application was given by Vasy in \cite{V}
where microlocal complex absorbing potentials were used to obtain Fredholm
properties and meromorphic continuation of the resolvents (see also 
\cite[Chapter 5]{res}). 

We conclude this section with some examples to which Theorem \ref{t:1} does {\em not} 
apply directly but which fit in the same framework. 

\medskip
\noindent
{\bf Example 1.} As explained in \cite[(c.31)--(c.33)]{SjDuke} the theory 
of Helffer--Sj\"ostrand \cite{He-Sj0} applies to the case of potentials
which are homogeneous of degree $ m $ and satisfy the condition
$  V ( x ) = 0 , \  x \neq 0  \Longrightarrow   \nabla V ( x ) \neq 0 $.
That means that resonances of $ P = - \Delta + V $ can be defined 
in $ \{ z \in \CC , \arg z > - \theta_0  \}$ for some $ \theta_0 > 0 $. 
It is interesting to ask if the viscosity limit gives a global definition 
in that case. 

That is easily seen in the case of quadratic potentials.
In fact, suppose that 
\[ V ( x ) = \lambda_1^2 x_1^2 + \cdots + \lambda_r^2 x_r^2 - 
\mu_1^2 x_{r+1}^2 - \cdots - \mu_{n-r}^2 x_n^2 , \ \ \lambda_j, \mu_\ell > 0 .\]
As recalled in \S \ref{dho} the eigenvalues of $ P_\epsilon $, $ \epsilon > 0 $,
are given by 
\[  \sum_{ j=1}^r ( \lambda_j^2 - i \epsilon)^{\frac12} ( 2 k_j + 1 ) - 
i \sum_{j=1}^{n-r} ( \mu_j^2 - i \epsilon)^{\frac12} ( 2 k_{ j+r} + 1) , \ \ 
k\in \NN_0^n , \]
where the branch of the square root is chosen to be positive on $ \RR_+ $. 
As $ \epsilon \to 0+ $ we obtain the globally defined set of resonances:
\[  \sum_{ j=1}^r \lambda_j ( 2 k_j + 1 ) - 
i \sum_{j=1}^{n-r} \mu_j ( 2 k_{ j+r} + 1) , \ \ 
k\in \NN_0^n . \]

\medskip
\noindent
{\bf Example 2.} This example fits in the framework of black box
scattering with one dimensional infinity. Consider the modular
surface $ M = SL_2 ( \mathbb Z ) \backslash \mathbb H^2 $ and $ \Delta_M \leq 0 $
the Laplacian  on $ M $. We then put $ P = - \Delta_M - \frac14 $ where
$ \frac14 $ guarantees that the the continuous spectrum of $ P$ is given 
$ [ 0  ,\infty ) $. This is a {\em black box} Hamiltonian 
on $ \mathcal H_0 \oplus L^2 ( [ 0 , \infty ) ) $ in the sense of
 \cite{SZ1} -- see 
\cite[\S 4.1]{res}. Traditionally, the resonances of the quotient 
$ M $ are defined
as poles of the meromorphic continuation of $ ( - \Delta_M - s ( 1 - s ))^{-1} $
from $ \Re s > \frac12 $ to $ \CC$ and are given by the embedded 
eigenvalues with $ \Re s = \frac12$ and by the non-trivial zeros of 
$ \zeta ( 2 s ) $ where $ \zeta $ is the Riemann zeta function. 
The resonances of $ P $ are then given by $ s (  1 - s ) $. 
 
If we choose the fundamental domain of $ SL_2 ( \mathbb Z ) $ to be 
$ \{ x + iy  : |x| \leq 1, x^2 + y^2 \geq 1 \} $ then the Laplacian 
in the cusp $ y > 1 $ is $ y^{-2}( \partial_x^2 + \partial_y^2 ) $.
The Hamiltonian on $ L^2 ( [ 0, \infty )_r ) $ is given by 
$ - \partial_r^2 $, $ r = \log y $ -- see \cite[\S 4.1, Example 3]{res}. 
In the language of Theorem~\ref{t:1} (see Remark 5) and in $ ( x, y ) $ coordinates
\[  P_\epsilon = - \Delta_M - \textstyle{\frac14} - i \epsilon ( 1 - \chi ( y ) ) ( 
\log y )^2 , \]
where $ \chi \in \CIc ( [0, \infty ) ) $, $ \chi ( y ) \equiv 1 $ for $ y < \frac32$
and $ \chi (y ) \equiv 1 $ for $ y > 2 $. The operator $ P_\epsilon $ has 
discrete spectrum for $ \epsilon > 0 $ and the eigenvalues converge to 
the resonances of $ P $ uniformly on compact subsets of $ \Im z > - \pi /4 $. 
Equivalently if we define $ \Sigma_\epsilon $
\[   s ( \epsilon ) \in \Sigma_\epsilon \Longleftrightarrow s( \epsilon ) ( 1 - 
s ( \epsilon ) ) \in \sigma ( P_\epsilon ) \]
The limit points of $ \Sigma_\epsilon $, $ \epsilon \to 0+ $, in 
$ \Re s < \frac12 $, $ |s | > C $ are given by the nontrivial 
zeros of $ \zeta ( 2s ) $.

\medskip
\noindent
{\bf Example 3.} Suppose that $ X $ is a compact manifold and $ V$ is 
a vector field on $ M $ generating an Anosov flow, $ \varphi^t = \exp t V $.
That means that the tangent space
to $ X $ has a continuous decomposition $ T_x X = E_0 ( x )                    
\oplus E_s (x) \oplus E_u ( x)$ which is invariant,
$ d \varphi_t ( x ) E_\bullet ( x ) = E_\bullet ( \varphi_t ( x ) ) $,
$E_0(x)=\mathbb R V(x)$,
and for some $ C$ and $ \theta > 0 $ fixed
\begin{equation}
\label{e:anosov}
\begin{split}
&  | d \varphi_t ( x ) v |_{\varphi_t ( x )}  \leq C e^{ - \theta              
  |t| } | v |_x  , \ \ v
\in E_u ( x ) , \ \
t < 0 , \\
& | d \varphi_t ( x ) v |_{\varphi_t ( x ) }  \leq C e^{ - \theta |t| } |  v |_x , \ \
v \in E_s ( x ) , \ \
t > 0 .
\end{split}
\end{equation}
where $ | \bullet |_y $ is given by a smooth Riemannian metric on $            
X $. A class of examples is given by $ X = T^1 M $ where $ M $ is 
a negatively curved Riemannian manifold and $ \varphi^t $ is the 
geodesic flow in its unit tangent bundle $ X $. 

If $ \Delta_g \leq 0 $ is the Laplacian for some metric on $ X $ 
then -- see \cite{dam} --  the limit set of the 
spectrum of 
\[   P_\epsilon = X/i + i \epsilon \Delta_g , \ \ \epsilon \to 0 + \]
is a discrete set given by the {\em Pollicott--Ruelle} resonances -- see
\cite{dam} for the definition and references. Adding the Laplacian 
corresponds to taking a {\em viscosity} regularization and that
explains our terminology. Another interpretation is given in terms
of Brownian motion: the pullback by the flow 
flow $ x ( t ) := \varphi_{t} ( x ( 0 ) ) $, is given by
$  e^{  i t P_0 } f ( x) = f (x (  t)  ) $, 
$ \dot x ( t ) = - V_{ x( t ) }$, $ x(0)=x$. 
For $ \epsilon > 0 $ the evolution equation is replaced by the
Langevin equation:
\[  e^{ - i t P_\epsilon  } f ( x ) = \mathbb E \left[ f (x (  t)  ) \right] ,\
 \ \
\dot x ( t ) = - V_{ x( t ) } + \sqrt{ 2 \epsilon } \dot B( t )  , \quad x(0)=\
x,\]
where $ B ( t ) $ is the Brownian motion corresponding to the metric $ g$ 
on $ X$. Hence considering
$ P_\epsilon $ corresponds to a stochastic perturbation of the deterministic flow.
In the case of scattering resonances the same interpretation can be proposed
on the Fourier transform side.

The assumption that the flow satisfies \eqref{e:anosov} is crucial as
otherwise the limit set is typically not discrete. The simplest
example is given by $ X = \mathbb S^1 \times \mathbb S^1 $,  $ \mathbb S^1 = 
\mathbb R / 2 \pi \mathbb Z $, 
and $ V = \partial_{x_1} + \alpha \partial_{x_2} $, $ \alpha \notin 
\mathbb Q $, $ \Delta_g = \partial_{x_1}^2 + \partial_{x_2}^2 $. In that
case the limit set of the spectrum of  $P_\epsilon $ is the lower half plane.
Other limit sets are possible, for instance in the case of the geodesic flow on 
$ \mathbb S^2 $, $ X = T^1 \mathbb S^2 \simeq SO(3)$. 
The spectrum of $ P_0 $ is given by $ \mathbb Z$
(with infinite multiplicities) and 
if we take $ \Delta_g $ 
to be the Casimir operator then the limit set of the spectrum of 
$ P_\epsilon $ as $ \epsilon \to 0 +$ is $ \mathbb Z - i [ 0 , \infty ) $.
For yet another example see \cite[\S 1]{dam}. 

\medskip
\noindent
{\bf Example 4.} We expect that viscosity definition of resonances 
remains valid, in a small angle near the real axis, for all 
{\em dilation analytic} potentials -- see \cite{He-Sj0} and references given there
and \S \ref{cs} below for a review of complex scaling. It would be interesting
to find a Schr\"odinger operator $ P$  for which the limit set of the spectrum of $ P_\epsilon $, $ \epsilon \to 0 $ is not discrete. Candidates are given by potentials
which are {\em not} dilation analytic, for instance, 
\[   - \partial_x^2 + \frac{ \sin x } x , \ \ x \in \RR . \]

\medskip
\noindent
{\bf Notation.}
We use the following notation: $ f =  \mathcal O_\ell (
 g   )_H $ means that
$ \|f \|_H  \leq C_\ell  g $ where the norm (or any seminorm) is in the
space $ H$, and the constant $ C_\ell  $ depends on $ \ell $. When either 
$ \ell $ or
$ H $ are absent then the constant is universal or the estimate is
scalar, respectively. When $ G = \mathcal O_\ell ( g ) : {H_1\to H_2 } $ then
the operator $ G : H_1  \to H_2 $ has its norm bounded by $ C_\ell g $.
Also when no confusion is likely to result, we denote the operator 
$ f \mapsto g f $ where $ g $ is a function by $ g $.

\medskip
\noindent
{\sc Acknowledgments.} The author would like to thank
Mike Christ, Semyon Dyatlov, Jeff Galkowski, John Strain and Joe Viola for helpful discussions.
This project was supported in part by 
the National Science Foundation grant DMS-1201417.

\section{Review of complex scaling}
\label{cs}

The complex scaling method changes the original Hamiltonian 
$ P = P_0 $ to a non-self-adjoint Hamiltonian $ P_{0,\theta }$ 
such that $  P_{0,\theta} - z : H^2 \to L^2 $ is a Fredholm 
operator when $ \arg z > - 2 \theta $. It was introduced 
by Aguilar--Combes \cite{AgCo}, Balslev--Combes \cite{BaCo} and Simon 
\cite{Si1}. For a review of practical applications of this method
in computational chemistry see Reinhardt \cite{Rei}. As the method
of {\em perfectly matched layers} (PML) it has reappeared in 
numerical analysis -- see Berenger \cite{ber}. The presentation 
here follows the geometric approach of Sj\"ostrand--Zworski \cite{SZ1}.
Eventually the proof that the viscosity eigenvalues converge to 
scattering resonances is a straightforward application of the
methods of \cite{SZ1} (see also \cite[\S 7.2]{SjR} for a more detailed
presentation and \cite[\S 4.5]{res} for an approach to complex scaling
based on the continuation of the Green function $ G ( z , x , y )$ in 
\eqref{eq:Green} in variables $ x $ and $ y $).

Suppose that $ \Omega \subset \CC^n $ is an open subset and that 
\begin{equation}
\label{eq:PzDz} 
 P ( z, D_z ) = \sum_{ |\alpha | \leq m } a_\alpha ( z ) D_z^\alpha, \ \ 
D_{z_j} := \textstyle{\frac 1 i } \partial_{z_j} , \ \ D_z^\alpha = 
D_{z_1}^{\alpha_1} \cdots D_{z_n}^{\alpha_n } , \end{equation}
is a differential operator with holomorphic coefficients. For instance
we can have $ P ( z , D_z ) = \sum_{ j=1}^n D_{z_j}^2 - i \epsilon z_j^2  $.

Suppose that $ \Omega \subset \CC^n $ is an open subset and that $ \Gamma \subset \Omega $ is a {\em maximal totally real}
submanifold. That means that $ \Gamma $ is
 a smooth real submanifold of dimension $ n $ such that
\begin{equation}
\label{eq:totre} \forall \, x \in \Gamma , \ \    T_x \Gamma \cap i T_x \Gamma = \{ 0 \} .\end{equation}
Here we identify $ T_x \Gamma $ with a real subspace of $ \CC^n $. The condition
\eqref{eq:totre} means that there exists a {\em complex linear} change of variables $ A : \CC^n 
\to \CC^n $ such that $ A ( T_x \Gamma ) = \RR^n \subset \CC^n $. Locally, 
$ \Gamma $ can be represented using real coordinates:
\begin{equation} 
\label{eq:fcoord}  \RR^n \supset U \ni x \mapsto f ( x ) = ( f_1 ( x ) , \cdots , 
f_n ( x ) )  \in \Gamma \subset \Omega \subset \CC^n .\end{equation}
Composing the matrix $  \partial_x f ( x ) := (\partial_{x_j} f_k ( x ) )_{ 1 \leq  k,j\leq n } $ with 
$ A $ we obtain an invertible matrix $ \RR^n \to \RR^n $. That means that
\begin{equation}
\label{eq:detf}
\det \left( \frac{ \partial f_k ( x ) } {\partial x_j } \right)_{ 1 \leq k , j \leq n }
\neq 0 . 
\end{equation}
Conversely, if \eqref{eq:detf} holds, then $ \partial f ( x ) $ is an 
injective complex linear matrix and for any sets $ U , V \subset \CC^n $, 
$ \partial f ( x ) ( U ) \cap
\partial f ( x ) ( V ) = \partial f ( x ) ( U \cap V ) $. 
Hence, 
\[  \begin{split} T_x \Gamma \cap i T_x \Gamma & = \partial f ( x )  ( \RR^n ) \cap 
i  \partial f ( x ) ( \RR^n ) = \partial f ( x ) ( \RR^n ) \cap 
\partial f ( x ) ( i \RR^n ) \\
& = \partial f ( x ) ( \RR^n \cap i \RR^n ) = \{ 0 \} ,\end{split} 
\]
and \eqref{eq:detf} implies \eqref{eq:totre}.
The volume form on $ \Gamma $ is obtained by pushing forward the standard volume
form on $ \RR^n $ by $ f $. That of course depends on the choice of $ f $ 
(in what follows the uniformity will be guaranteed by \eqref{fth2} below).

\medskip
\noindent
{\bf Example. } As a simple illustration consider $ n=2 $ and $ f ( x_1 , x_2 ) = ( x_1 + i x_2 , 0 ) \in \CC^2 $. Then
\[ \partial_x f ( x ) = \begin{bmatrix} 1 & i \\ 0 & 0 \end{bmatrix}, \ \ 
T_x f ( \RR^2 ) = \CC \oplus \{ 0 \} \subset \CC^2 . \]
The tangent space is not totally real and condition \eqref{eq:detf} is violated.
To introduce the next topic we also note we cannot restrict operators, $ P $, 
with holomorphic coefficients to $ f ( \RR^2 ) $ in a way that 
for holomorphic functions, $ u $,  $ (P u) |_{ f ( \RR^2) } = (P_{ f ( \RR^n ) } ) 
( u |_{f ( \RR^n ) } ). $ As an example consider $  P = \partial_{z_2} $ 
and $ u = z_2 $.

\medskip
The point of introducing totally real submanifolds $ \Gamma $ is the fact that
an operator, $ P $,  with holomorphic coefficients can be restricted to an 
operator with complex smooth coefficients on $ \Gamma $, $ P_\Gamma $, 
in such a way that for $ u $ holomorphic near $ \Gamma $, 
$ Pu |_\Gamma = P_\Gamma ( u |_\Gamma ) $.

The differential operator $ P ( z, D_z ) $ given in \eqref{eq:PzDz} defines a unique $ P_\Gamma $ a differential operator on $ \Gamma $ as follows.
Using \eqref{eq:fcoord} we can identify a small neighbourhood of any 
$ z_0 \in \Gamma $ with $ U \subset \RR^n $. Then $ u \in \CI ( \Gamma 
\cap f ( U ) ) $ can be identified with $ u \circ f \in \CI ( U ) $. We then have
\begin{equation}
\label{eq:Pgamma} ( P_\Gamma   u )  \circ f  (x) = 
\sum_{ |\alpha | \leq m } (a_\alpha \circ f ) ( x ) ( ({}^t \partial_x f ( x) ^{-1} 
D_{ x })^\alpha ( u \circ f )  ( x) . \end{equation}
It is easy to see that this definition is independent of the choice of  $f $
and that the condition \eqref{eq:detf} is crucial.

The key fact is the standard result about continuation of solutions to $ 
P_\Gamma u $. The proof based on \cite{L}, \cite{M} and \cite{S1} 
can be found in \cite[Lemma 3.1]{SZ1} and (in more detail) \cite[Lemma 7.2]{SjR}. 
With the notation above we have the following:

\begin{lem}
\label{l:ext}
Suppose that $ W \subset \RR^n $ is open and that 
$ F : [ 0 , 1 ] \times W \ni ( s, x ) \mapsto F (  s ,x ) \in \CC^n , $
is a smooth proper map satisfying for all $ s \in [ 0 , 1 ] $ 
\[  \det \partial_x F ( s , x ) \neq 0 , \ \  \text{and  $ x \mapsto F ( s , x )  $ is injective.} \]
In addition assume that there exists a compact set $ K \subset W $ such that 
\[  x \in W \setminus K  \  \Longrightarrow \ F ( 0 , x ) = F ( s ,x ) , \ \ 0 \leq s \leq 1 , \]
and that $ F ( [ 0 , 1 ] \times W ) \subset \Omega $ with $ P ( z , D_z ) $ 
a differential operator with holomorphic coefficients in $ \Omega $. 

Now assume that for $ \Gamma_s := F ( \{ s \} \times W ) $, 
$ P_{ \Gamma_s } $ is an elliptic differential operator in the sense that
\[    \big | \sum_{|\alpha| = m } a_\alpha ( z)  \zeta^\alpha \big | \geq
C | \zeta |^m , \ \ ( z , \zeta ) \in T^* \Gamma_s . \]

If $ u_0 \in \CI ( \Gamma_0 ) $ and $ P_{\Gamma_0} u_0 $ extends to a holomorphic
function on $ \Omega $, then for every $ s \in [ 0 , 1 ] $ there exists 
a holomorphic function, $ U_s $ defined near $ \Gamma_s $ such that, for some~$ \epsilon $, 
\[ U_0 |_{\Gamma_0 } = u_0 ,  \ \   |s - s' | < \epsilon  \ \Longrightarrow \ \text{ $ U_{s} = U_{s'} $ on the intersection of their domains.}
 \]
\end{lem}

In other words, the function $ u_0 $ defined on $ \Gamma_0 $ extends to a 
possibly multivalued function $ U $ in a neighbourhood of $ f ( [0,1] \times W ) $.

The lemma will be applied to a family of deformations of $ \RR^n $ in $ \CC^n $.
Our goal is to restrict the operator $ P_\epsilon = - \Delta - i \epsilon x^2 + V $, 
$ \epsilon \geq 0 $, to the corresponding totally real submanifolds. For that
the deformation has to avoid the support of $ V $ and we choose $ r_0 $
such that $ \supp V \subset B ( 0 , r_0 )$. We then construct 
\begin{equation}
\label{eq:defgt}  [ 0 , \pi ) \times [ 0 , \infty ) \ni ( \theta, t ) \longmapsto g_\theta ( t ) \in \CC \end{equation}
which is $ \CI $, is injective on $ [ 0 , \infty ) $ for every fixed $\theta$
and satisfies
\begin{gather}
\text{$ g_\theta ( t ) = t  $ for $ 0 \leq t \leq r_0 $,}
\label{fth1}\\
0 \leq \arg g_\theta ( t ) \leq \theta , \ \ \partial_t g_\theta \neq 0 , 
\label{fth2}\\
\arg g_\theta ( t ) \leq \arg \partial_t g_\theta ( t ) \leq \arg g_\theta ( t) + 
\epsilon_0 , 
\label{fth21}
\\
\text{ $g_\theta ( t ) = e^{i\theta } t $ for $ t \geq T_0 $ where $ T_0 $
depends only on $ \epsilon_0 $ and $ r_0 $.} 
\label{fth3}
\end{gather}
We now define the totally real submanifolds, $ \Gamma_\theta $, as images of
$ \RR^n $ under the maps
\begin{equation}
\label{eq:defGth}   f_\theta : \RR^n \to \CC^n , \ \  f_\theta ( x ) := g_\theta ( |x|) { x}/{|x|}
, \ \ \Gamma_\theta := f_\theta ( \RR^n ) . \end{equation}
For $ \epsilon \geq 0 $ and $ 0 \leq \theta < \pi $ we put
\begin{gather} \label{eq:defP}  \begin{gathered} 
- \Delta_\theta := (\Delta_z) |_{\Gamma_\theta } , \ \ \ \ x_\theta := z|_{\Gamma_\theta}, \\ Q_{\epsilon, \theta } : =
 - \Delta_\theta - i \epsilon x_\theta^2 , \ \ \ \ 
 P_{\epsilon, \theta } :=  Q_{\epsilon , \theta } + V . 
\end{gathered}
\end{gather}
Parametrizing $ \Gamma_\theta $ by $ ( t, \omega ) \in [0, \infty ) \times \mathbb
S^{n-1} $, $ ( t, \omega) \mapsto g_\theta ( t ) \omega $, we have
\begin{equation}
\label{eq:Delth}  -\Delta_\theta = \left( g_\theta ' ( t )^{-1} D_t \right)^2 - 
i (n-1) g_\theta(t)^{-1} g'_\theta(t)^{-1} D_t + g_\theta(t)^{-2} D_\omega^2 , 
\end{equation}
where $ D_t = \partial_t/i $ and $ D_\omega^2 = - \Delta_{ \mathbb S^{n-1} } $.
The symbol is given by 
\[ \sigma ( - \Delta_\theta ) =  g'_\theta(t)^{-1} \tau^2 + g_\theta ( t )^{-2} 
w^2 , \ \ \  ( t , \omega ; \tau , w )\in T^* ( [ 0 , \infty ) \times \mathbb S^{n-1}) . \]

The basic result based on ellipticity at infinity is 
\[ -2 \theta + \delta < \arg z <  2 \pi - 2 \theta - \delta , \ \ |z| \geq \delta \ \
\ \Longrightarrow  \  \ ( - \Delta_\theta - z )^{-1} = \mathcal O_\epsilon ( 1 ) : 
L^2 ( \Gamma_\theta ) \to H^2 ( \Gamma_\theta ) . \]
This follows from \cite[Lemmas 3.2--3.5]{SZ1} applied with $ P = - \Delta $. 
As will be reviewed in \S \ref{mc} this shows that $  P_{0,\theta} - z : H^2 
\to L^2 $ is a Fredholm operator in this range of values of $ z $ and that
the eigenvalues are independent of $ \theta $. 

The crucial property is 
\begin{lem}
\label{l:2}
Let $ R_0 ( z ) = ( - \Delta -z )^{-1} : L^2  \to H^2 $, 
$ \Im z > 0 $,
be the free resolvent and let $ R_0 ( z ) $ also denote its analytic continuation
across $ [ 0 , \infty ) $ as an operator $ L^2_{\rm{comp} } \to H^2_{\rm{loc}} $.

Suppose that $ \chi \in \CIc ( B ( 0 , r_0 ) ) $ so that $ \chi $ is 
defined on $ \Gamma_\theta $. Then for $ - 2 \theta < \arg z < 2 \pi - 
2 \theta $, $ \theta < \pi $, 
\begin{equation}
\label{eq:l2}  \chi R_0 ( z ) \chi = \chi ( - \Delta_\theta - z )^{-1} \chi. 
\end{equation}
\end{lem}
\begin{proof} We recall the main features of the proof which is
implicit in \cite[\S 3]{SZ1}. It is sufficient to establish the identity
\eqref{eq:l2} for $ 0 < \arg z < 2 \pi - 2 \theta $ 
as it then follows by analytic continuation. It is also enough to show that in 
this 
range of $  z$ and $ 0 \leq \theta_1 < \theta_2 \leq \theta$, $ |\theta_1 -
\theta_2 | \ll 1 $, 
\begin{equation}
\label{eq:l21}  \chi ( - \Delta_{\theta_1} - z)^{-1} \chi = \chi ( - \Delta_{\theta_2}  - z )^{-1} \chi. 
\end{equation} 
For that we show that for $ f \in L^2 ( B ( 0 , r_0 ) ) 
\subset L^2 ( \Gamma_{\theta_j} ) $  there exists $ U $ holomorphic in a 
neighbourhood $ \Omega_{\theta_1, \theta_2 } $ of 
$$ \bigcup_{ \theta_1 \leq \theta \leq \theta_2 } (  \Gamma_\theta 
\setminus B ( 0 , r_0 ) ) 
\subset
\CC^n $$
such that 
\begin{equation}
\label{eq:UGa}  U|_{\Gamma_{\theta_j } } (x )  = [ ( -\Delta_{\theta_j}  - z)^{-1} \chi f ] ( x ) 
\ \text{ for $ x \in \Gamma_{\theta_j } \setminus  B( 0 , r_0 ) $.} \end{equation}
The unique continuation property for second order elliptic 
operators then shows that 
\[  \chi ( \Delta_{\theta_1 }- z )^{-1} \chi f = 
\chi ( \Delta_{\theta_2} - z )^{-1} \chi f , \]
proving \eqref{eq:l2}.

\begin{figure}
\includegraphics{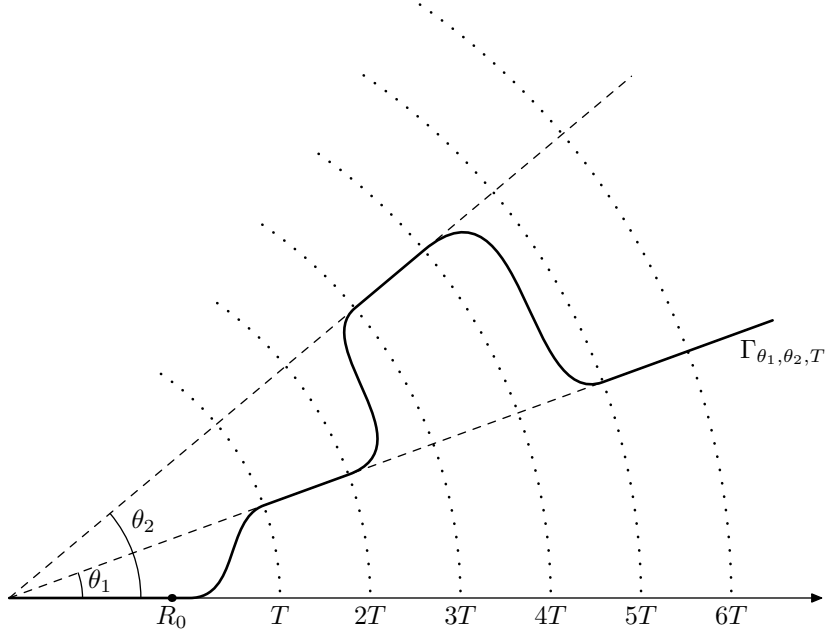}
\caption{The deformed totally real submanifold $ \Gamma_{\theta_1, \theta_2, T } $ interpolating between $ \Gamma_{\theta_1} $ and $ \Gamma_{\theta_2}$. }
\label{f:1}
\end{figure}

To show the existence of $ U $ such that \eqref{eq:UGa}  holds we
use Lemma \ref{l:ext} applied to a modified family of deformations.
The key is to show that a holomorphic extension, $ U $, of the solution to 
$ ( - \Delta_{\theta_1 } - z ) u_1  = \chi f $, $ u_1 \in L^2 ( \Gamma_{\theta_1} ) $, 
restricts to $ u_2 \in L^2 ( \Gamma_{\theta_2} )$ (the equation 
$ ( - \Delta_{\theta_2 } - z ) u_2 = \chi f $ is automatically satisfied). 
That means that $ u_2 = ( - \Delta_{\theta_2 } - z)^{-1} ( \chi f ) $ 
proving \eqref{eq:UGa}. 

The modified family of contours is obtained as follows. Fix $ T \gg 1 $ 
and choose $ \chi \in \CIc ( ( 2 , 5 )  ; [ 0 , 1 ])$
equal to $ 1 $ near $ [ 3, 4]  $. Then define
\begin{gather*}  g_{\theta_1, \theta_2 , T } ( t ) := 
g_{\theta_1 } ( t ) + \chi ( t/T) ( g_{\theta_2 } ( t ) - g_{\theta_1 } ( t ) ) , \\ 
\Gamma_{\theta_1 , \theta_2, T } := \{ g_{\theta_1, \theta_2 , T } ( t ) \omega \, :\,
t \in [ 0 , \infty ) , \ \omega \in \mathbb S^{n-1} \} \subset \CC^n  .\end{gather*}
We can apply Lemma \ref{l:ext} to the family of totally real submanifolds 
interpolating between $ \Gamma_{ \theta_1 } $ and $ \Gamma_{\theta_1, \theta_2, T} $:
$ [ 0 , 1 ] \ni s \longmapsto \Gamma_{ \theta_1 , \theta_1 + s ( \theta_1 , \theta_2), T } $.
That implies that there exists a holomorphic function $ U^T $ defined in 
a neighbourhood of the union of these submanifolds and such that $ u_1 =  
U^T |_{ \Gamma_{\theta_1 } } $. Changing $ T $ we obtain a family of functions
agreeing on the intersections of their domains and that gives  $ U $ defined in 
the neighbourhood $ \Omega_{ \theta_1, \theta_2 } $. To see that 
$ U |_{\Gamma_{\theta_2}} \in L^2 ( \Gamma_{\theta_2} ) $ it suffices to show that
\begin{equation} \label{eq:UTG} \| U^T |_{ \Gamma_{ \theta_1, \theta_2 , T } } \|_{ L^2 ( \Gamma_{\theta_1, 
\theta_2 , T} ) } \leq  C_0 \| u_1  \|_{L^2 ( \Gamma_{\theta_1} \cap \{  T \leq |z | \leq 6 T \} , ) }  , \end{equation}
where $ C_0 $ is independent of $ T $. (We apply \eqref{eq:UTG} 
with $ T = 2^j $ and sum over $ j $.)

To see \eqref{eq:UTG}
\begin{gather*} \Omega_1 ( T )  = \{ z \in \CC^n :  2 T \leq | z | \leq 5 T \} \cap
\Gamma_{ \theta_1, \theta_2, T } 
\supset \Gamma_{\theta_1, \theta_2, T } 
\setminus \Gamma_{\theta_1 }  , \\ 
\Omega_2 ( T )  = \{ z \in \CC^n :   T \leq | z | \leq 6 T \} \cap
\Gamma_{ \theta_1, \theta_2, T } , \ \ 
\Omega_2 ( T ) \setminus \Omega_1 ( T ) \subset e^{ i \theta_1 } \RR^n . \end{gather*}
We claim that for $ T $ large and $ u \in \CI ( \Gamma_{\theta_1, \theta_2, T } )$,
\begin{equation}
\label{eq:utt}   \| u \|_{ L^2 (\Omega_1 ( T ) ) } \leq C \| ( - \Delta_{ \Gamma_{\theta_1, 
\theta_2 , T } } - z ) u \|_{ L^2 ( \Omega_2 ( T ) )}  + 
C \| u \|_{ L^2 ( \Omega_2 ( T ) \setminus \Omega_1 ( T ) )} . \end{equation}
For $ | \theta_2 - \theta_1 | \ll 1 $, this estimate is a perturbation 
of a standard semiclassical elliptic estimate:
treating $ h:= 1/T $ as a semiclassical parameter, uniform ellipticity of 
$ - e^{ -2 i \theta } h^2 \Delta - z $ shows that for $ v \in \CI ( \RR^n ) $, 
\[   \| v \|_{ L^2 ( \{2 \leq |x| \leq 5 \})} \leq
C \| ( - e^{-2i \theta}  h^2 \Delta - z ) v \|_{ L^2 ( \{1 \leq |x| \leq 6 \}) } 
+ C \| v \|_{ L^2 ( \{ 1 \leq |x| \leq 2 \} \cup \{ 5 \leq |x| \leq 6 \} } .\]
(This can be seen applying the inverse from \cite[Theorem 4.29]{e-z} to $  \chi v $
where $ \chi \in \CIc ( ( 1 , 6 )) $ is equal to $ 1 $ on $ [ 2, 5 ] $.)
The properties of $ \Omega_j ( T ) $ then imply \eqref{eq:UTG} completing the
argument.
\end{proof}

\section{The Davies harmonic oscillator}
\label{dho}

The operator $ H_{\epsilon, \gamma} := - \Delta + e^{ - i \gamma } \epsilon x^2 $,  
$ \epsilon > 0 $, $0 \leq \gamma < \pi $, 
was used by Davies \cite{Dav} to illustrate properties of 
non-normal differential operators. We recall the following basic result:

\begin{lem}
\label{l:3} 
The operator $ H_{\epsilon , \gamma } $ is an unbounded operator 
on $ L^2 $ with the discrete spectrum given by 
\begin{equation}
\label{eq:l30} \sigma ( H_{\epsilon , \gamma } ) = e^{ - i\gamma/2} \sqrt \epsilon ( n + 2|\NN_0^n| ), \ \ | \alpha | = \alpha_1 + \cdots + \alpha_n . \end{equation}
If $ \Omega \Subset \{ z :  -\gamma < \arg z < 0 \} \setminus 
e^{- i \gamma/2  } [ 0 , \infty )  $, 
then for some constant $ C_1 =  C_1 ( \Omega) $, 
\begin{equation}
\label{eq:l31}   \frac{1}{ C_1} e^{ \epsilon^{-\frac12} /C_1 } \leq
\| ( H_{ \epsilon , \gamma } - z )^{-1} \|_{ L^2 \to L^2 } \leq 
{ C_1 }  e^{ C_1 \epsilon^{-\frac12} } \ \ , z \in \Omega.   \end{equation}

In addition for any $ \delta > 0 $ there exists a constant $ C_2 $ such that,
uniformly in $ \epsilon > 0 $, 
\begin{equation}
\label{eq:l32}
\|  ( H_{ \epsilon , \gamma } - z )^{-1} \|_{ L^2 \to L^2 }  \leq C_2/ |z| , \ \ 
 \delta < \arg z < 2 \pi - \gamma - \delta, \ \ |z| > \delta . 
 \end{equation}
 \end{lem}
\begin{proof}
By rescaling $ y = \sqrt \epsilon x $
this operator in unitarily equivalent to $ - \epsilon \Delta_y + e^{-i \gamma}   y^2 $, 
that is a semiclassical, $ h = \sqrt \epsilon $, quadratic operator. 
For the analysis of the spectrum and upper bounds on the resolvent 
for general quadratic operators see Hitrik--Sj\"ostrand--Viola \cite{HSV} 
and references given there
-- in particular we obtain \eqref{eq:l30} and the upper bound in \eqref{eq:l31}. The lower bound in \eqref{eq:l31} 
follows from general arguments for operators with analytic 
coefficients -- see \cite[\S 3]{DSZ} and the bound \eqref{eq:l32} from 
(semiclassical) ellipticity of $ - h^2 \Delta_y + e^{ - i \gamma } y^2 - z $
for $ \delta < \arg z < 2 \pi - \gamma - \delta $, $ |z| > \delta $.
\end{proof}

\begin{figure}
\includegraphics{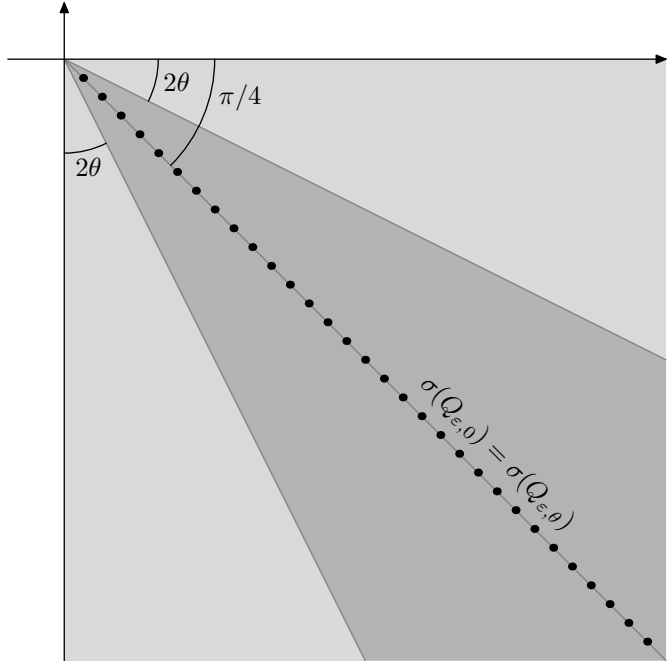}
\caption{A visualization of 
the spectrum of $ Q_{\epsilon , 0 } = - \Delta - i \epsilon x^2 $ 
which is equal to the spectrum of the deformed operator $ Q_{\epsilon , 
\theta } $. The lightly shaded region is the numerical range 
of $ Q_{\epsilon, 0 } $ and the darker shaded region, the numerical 
range of $ - e^{ -2 i \theta} \Delta - i e^{ 2 i \theta} \epsilon  x^2 $. 
The estimates for the resolvents of $ Q_{\epsilon, \theta} $ improve
outside of that region.} 
\label{f:2}
\end{figure}

We now consider the special case of $ H_{ \epsilon , \pi/2 } = 
Q_{\epsilon , 0 } $  and of its deformation 
$ Q_{\epsilon , \theta } $ -- see \eqref{eq:defP}.
The facts we need are given in the next two lemmas.
The first is the analogue of Lemma \ref{l:2}:
\begin{lem}
\label{l:5}
In the notation of Lemma \ref{l:2}, $ 0 \leq \theta \leq \pi/8 $, 
$ \epsilon > 0 $,  and
$  -2 \theta < \arg z < 3 \pi /2 + 2 \theta $ we have 
\begin{equation}
\label{eq:l5}
\chi ( Q_{ \epsilon , 0 } -z )^{-1} \chi = 
\chi ( Q_{\epsilon,  \theta } - z )^{-1} \chi . 
\end{equation}
In particular, for $ 0 \leq \theta \leq \pi/8 $, the spectrum is
independent of $ \theta $ and given by $ \sqrt \epsilon e^{ - i \pi/4}
( n + 2|\mathbb N_0^n|) $.
\end{lem}
\begin{proof} 
We follow the argument in the proof of Lemma \ref{l:2} and use
the notation introduced there. Hence it is enough to prove that
$ 0 \leq \theta_1 < \theta_2 \leq \pi/8 $ and $ |\theta_1 - \theta_2 | $
small it is enough to show that
\[ \chi ( Q_{ \epsilon , \theta_1 } -z )^{-1} \chi = 
\chi ( Q_{\epsilon,  \theta_2 } - z )^{-1} \chi .\]
We only need to establish this for $ z \in e^{  i ( - 2\theta_1 + \pi /2 )} ( 1, \infty ) $ as 
then the result follows by analytic continuation. 
 The only difference is an estimate which replaces \eqref{eq:utt}: for $ \tau > 1 $, 
\begin{gather}
\label{eq:zutt} 
\begin{gathered} 
\| u \|_{ L^2 (\Omega_1 ( T ) ) } \leq  C \| ( Q_{ \Gamma_{\theta_1, 
\theta_2 , T } }  - ie^{ - 2 
\theta_1 } \tau ) 
u \|_{ L^2 ( \Omega_2 ( T ) )} + C \| u \|_{ L^2 ( \Omega_2 ( T ) \setminus \Omega_1 ( T ) )} , \\
Q_{\theta_1, 
\theta_2 , T } :=  - \Delta_{ \Gamma_{\theta_1, 
\theta_2 , T } } - i \epsilon (x|_{\Gamma_{\theta_1, \theta_2, T } })^2 
\end{gathered} \end{gather}
uniformly for $ T \gg 1 $. 
To see this we first note that for $ \epsilon > 0 $, 
$  Q_{ {\theta_1, 
\theta_2 , T } } - z  $, $ z \in \CC $, 
is a Fredholm operator (since it is elliptic and near infinity it is equal to 
$ e^{ -2 i \theta } H_{\epsilon , \pi/2 - 4 \theta } $). 

To obtain an estimate we notice that 
for $ t > T $
\[ g_{\theta_1, \theta_2, T}' ( t ) = \chi ( t/T ) e^{i\theta_2 } + 
( 1 - \chi ( t/T ) ) e^{ i \theta_2 } + ( t/T ) \chi' ( t/T ) ( e^{ i \theta_2} -
e^{ i \theta_2 } )  ,\]
so that from \eqref{fth2} and \eqref{fth3},
\[  \theta_1 - C | \theta_2 - \theta_1| \leq \arg g_{\theta_1 , \theta_2, T } ' ( t ) 
\leq \theta_2 .\]

Also, $ \theta_1 \leq  \arg g_{\theta_1, \theta_2, T } ( t ) \leq \theta_2 $.
Hence,
\[  \Re \langle  
( e^{  2 i \theta_1 } Q_{ {\theta_1, 
\theta_2 , T } }  - i \tau) 
u , u \rangle \geq  \| D u \|^2 
/C 
\]
where we used the fact that for for $ 0 \leq \theta 
\leq \pi/8 $, $\Re ( - i e^{ 4 \theta } ) \geq 0 $. 
The imaginary part is then estimated as follows, 
\[   \begin{split}  - \Im \langle (e^{  2 i \theta_1 }
 Q_{ {\theta_1, 
\theta_2 , T } }  - i \tau ) u , 
u \rangle &  \geq \tau \| u \|_{ L^2 ( \Gamma_{\theta_1, \theta_2, T }  ) }
- \mathcal O ( | \theta_2 - \theta_1 |) 
\| D u \|^2 .\end{split} 
\]
We conclude that when $ |\theta_2 - \theta_1| $ is small enough 
\[ \| (Q_{ \Gamma_{\theta_1, 
\theta_2 , T } }  - i e^{ - 2 i \theta_1 } \tau ) u \|
\geq (\| u \|  +  \| D u \| ) / C ,  \]
This and the Fredholm property imply that 
\[  (Q_{{\theta_1, 
\theta_2 , T } }  - i e^{ - 2 i \theta_1 } \tau ) ^{-1} = \mathcal O ( 1 ) :
L^2 ( \Gamma_{ \theta_1, 
\theta_2 , T }  ) \to H^1 (\Gamma_{ \theta_1, 
\theta_2 , T }  )  . \]
that is the operator is invertible with bounds independent of $ T $.  
From this \eqref{eq:zutt} follows by a standard localization argument:
we choose $ \chi_T \in \CI ( \Omega_2 ( T )  , [ 0 , 1 ] ) $, 
such that $ \chi_T = 1 $ on $ \Omega_1 ( T ) $ with derivative bounds
independent of $ T $. We then apply 
the inverse above to $  (Q_{{\theta_1, 
\theta_2 , T } }  - i e^{ - 2 i \theta_1 } \tau ) \chi_T u $ with 
the commutator terms estimated by $ \| u \|_{ L^2 ( \Omega_2 ( T ) 
\setminus \Omega_1 ( T ) ) }$. 
\end{proof}

The next lemma shows how complex scaling dramatically improves
the exponential bound \eqref{eq:l31}:
\begin{lem}
\label{l:4}
Suppose that $ 0 \leq \theta \leq \pi/8 $ and that 
$ \Omega \Subset \{ z :  -2 \theta < \arg z < 3 \pi /2 + 2 \theta \} $. 
Then there exists $ C = C ( \Omega ) $ (in particular independent of $ \epsilon > 0 $)
such that 
\[  \| ( Q_{ \epsilon, \theta } - z )^{-1} \|_{ L^2 \to L^2 } \leq C , \ \ 
z \in \Omega . \]
\end{lem}
\begin{proof}
Let $ \chi_j \in \CIc ( [0, \infty ) ) $ be equal to $ 1 $ on $ [0, r_0 ] $
and satisfy $ \chi_j = 1 $ on $ \supp \chi_{j+1} $, $ j = 0,1.$
Parametrizing $ \Gamma_\theta $ by $ F_\theta: [ 0, \infty )_t
\times \mathbb S^{n-1} \to \Gamma_\theta $,  $ F_\theta ( t , \omega) = 
g_\theta ( t )\omega $ (with $ g_\theta $ given in \eqref{eq:defgt})
we define functions $ \chi_j^h 
\in \CIc ( \Gamma_\theta ) $ as 
\[  \chi_j^h \circ F_\theta ( t , \omega ) := \chi_j ( t h ) , \ \ 0 < h \leq 1 .\]
In view of \eqref{fth3} and \eqref{eq:Delth} we see that for $ h $ small
enough
\[ \begin{split}   Q_{\epsilon, \theta } ( 1 - \chi_1^h ) & = ( - e^{ -2 i \theta } \Delta_x - i 
\epsilon e^{ 2 i \theta } x^2 ) ( 1 - \chi_1^{h} ) \\
& = 
e^{ - 2 i \theta} H_{ \epsilon, \gamma } ( 1 - \chi_1^h ) , \ \ 
\gamma := \pi/2 - 4 \theta, \ \ \
x = t \omega . \end{split} \]
In view of \eqref{eq:l32} we have
\begin{equation}
\label{eq:Heps}  ( 1 - \chi_2^h )  e^{ 2 i \theta } ( H_{ \epsilon , \gamma } - e^{ 2 i \theta} z )^{-1} ( 1 - \chi_2^h ) = 
\mathcal O_\delta ( 1 ) : L^2 ( \Gamma_\theta ) \to H^2 ( \Gamma_\theta ) , \end{equation} 
for 
\[ - \delta < 2 \theta + \arg z < 2 \pi - \gamma - \delta = 3 
\pi/2  + 4 \theta - \delta , \ \ \ |z| > \delta ,  \]
and in particular for $ z \in \Omega $. We stress that the bounds are independent
of $ \epsilon $.

Noting that 
\begin{equation}
\label{eq:Deleps}   ( - \Delta_\theta - z )^{-1} = \mathcal O ( 1 ) : L^2 ( \Gamma_\theta) 
\to H^2 ( \Gamma_\theta ) ,  \ \ z \in \Omega , \end{equation}
(for $ 0 \leq \theta \leq \pi/8$, $ - 2 \theta < \arg z <  2 \pi - 2 \theta $)
we now put
\[  T_{\epsilon, \theta}^h ( z )  :=  \chi_0^h ( - \Delta_\theta - z)^{-1} \chi_1^h
+ ( 1 - \chi_1^h )  e^{ 2 i \theta } ( H_{ \epsilon , \gamma } - e^{ 2 i \theta} z )^{-1} ( 1 - \chi_2^h ) ,\]
so that
$   ( Q_{\epsilon , \theta } - z ) T_{\epsilon , \theta}^h ( z ) = 
I + K_{ \epsilon, \theta}^h ( z ) $,
where 
\[ \begin{split} K_{\epsilon, \theta }^h ( z ) := & - i \epsilon x_\theta^2 \chi_0^h 
( - \Delta_\theta -  z)^{-1} \chi_1^h - [ \Delta_\theta, \chi_0^h ] ( - \Delta_\theta -  z)^{-1} \chi_1 \\
& \ \ \ + [ \Delta_\theta, \chi_1^h] e^{ 2 i \theta } ( 1 - \chi_2^h) ( H_{ \epsilon , \gamma } - e^{ 2 i \theta} z )^{-1} ( 1 - \chi_2^h ) . 
\end{split} \]
Since $ [ \Delta_\theta , \chi^h_j] = \mathcal O ( h ) : H^1 ( \Gamma_\theta ) 
\to L^2 ( \Gamma_\theta ) $ and $ x^2_\theta \chi_1^h = 
\mathcal O ( h^{-2} ) : L^2 ( \Gamma_\theta ) \to L^2 ( \Gamma_\theta ) $, 
we conclude from \eqref{eq:Heps} and \eqref{eq:Deleps}
that for $ z \in \Omega $, 
\[ K_{\epsilon , \theta }^h  ( z ) = \mathcal O ( h^{-2} \epsilon ) +
\mathcal O ( h ) : L^2 ( \Gamma_\theta ) \to L^2 ( \Gamma_\theta ) .\]
Hence by choosing $ h $ first, we see that for $ \epsilon < \epsilon_0 ( h ) $, 
$  I + K_{\epsilon , \theta }^h( z)  $ has a uniformly bounded inverse
and $ 0 \leq \epsilon < \epsilon_0 $
\[ ( Q_{\epsilon , h }  - z )^{-1} = T_{\epsilon, \theta}^h ( z ) ( I + K_{\epsilon, 
\theta } ^h ( z ) )^{-1} = \mathcal O ( 1 ) : L^2 ( \Gamma_\theta ) \to 
L^2 ( \Gamma_\theta ) , \ \ \ z \in \Omega . \]
In view of Lemma \ref{l:5} we know that for $ z \in \Omega $, 
$ ( Q_{\epsilon , h }  - z )^{-1}$ exists for $ \epsilon > \epsilon_0 $ 
and that gives the bound for all values $ \epsilon $. 
\end{proof}

\section{Meromorphic continuation}
\label{mc}

In this section we will review the meromorphy of the resolvent $ R_V ( z ) $, see \eqref{eq:Rofz}, in a way connecting it to the resolvent of $ P_\epsilon $
given in \eqref{eq:Peps}, $ \epsilon \geq 0 $. 
For that we define
\begin{equation}
\label{eq:Reps}
R_\epsilon ( z ) = ( - \Delta - i \epsilon x^2 - z )^{-1} , \ \ 
R_{V , \epsilon } ( z ) = ( - \Delta - i \epsilon x^2 + V - z )^{-1} , \ \ 
\epsilon \geq 0 . 
\end{equation}
For $ \epsilon > 0 $, these operators are meromorphic for $ z \in \CC $
as operators on $ L^2 $. For $ \epsilon = 0 $, $ R_0 ( z ) $ is
holomorphic in the sense of \eqref{eq:Rofz} 
on the double cover of $ \CC \setminus \{ 0 \} $ when
$ n $ is odd and on the logarithmic cover when $ n $ is even 
-- see for instance \cite[\S 3.1]{res}. We are only concerned with 
continuation to $ \arg z \geq - \pi/4 $. 

Let $ \rho \in \CIc ( \RR^n ; [ 0 , 1 ] ) $ be equal to $ 1 $ on a 
neighbourhood of $ \supp V $. We have 
\begin{lem} 
\label{l:mult}
For $ \epsilon \geq 0 $ 
\[  z \mapsto  ( I + V  R_\epsilon ( z )\rho )^{-1} , \ \ \ 
  - \pi /4 < \arg z < 7 \pi/4 , 
\]
is a meromorphic family of operators on $ L^2 ( \RR^n ) $ for 
with poles of finite rank. 
Then 
\begin{equation}
\label{eq:Pieps}
m_\epsilon ( z ) := \frac{ 1}{ 2 \pi i} \tr
\oint_z   ( I + V R_\epsilon ( w ) \rho)^{-1}\partial_w (V R_\epsilon ( w ) 
\rho ) dw , 
\end{equation}
where the integral is over a positively oriented circle enclosing $ z $
and containing no poles other than possibly $ z $, satisfies \begin{equation}
\label{eq:lmult}
m_\epsilon ( z ) = \left\{ \begin{array}{ll}  \frac{1}{ 2 \pi i } \oint_z 
( w - P_\epsilon )^{-1} dw , &  \epsilon > 0 \\
\ & \ \\
\ \ \ \ \ \ m ( z ) , &  \epsilon = 0 , \end{array} \right. 
\end{equation}
where  $ m ( z ) $ is the multiplicity of the resonance $ z $
given by \eqref{eq:mz}. 
\end{lem}
\begin{proof}
We recall the standard argument (see \cite[\S 2.2, 3.2]{res} and 
references given there). For any 
$ \delta > 0 $ and uniformly in $ \epsilon \geq 0 $,  
\begin{equation}
\label{eq:Repsz}  R_\epsilon ( z ) = \mathcal O_\delta ( 1/|z| ) : L^2 ( \RR^n ) \to L^2 ( \RR^n ) , \ \ \ 
 \delta < \arg z < 3 \pi/2 - \delta , \ \ 
|z| > \delta .  \end{equation}
This follows from self-adjointness for $ \epsilon = 0 $ and 
from \eqref{eq:l32} for $ \epsilon > 0 $. 

For $ z $ in \eqref{eq:Repsz} and $ Q_\epsilon := - \Delta - i \epsilon x^2 $, 
\begin{equation}
\label{eq:RPeps}  \begin{split}  ( P_\epsilon - z ) & =  (  Q_{\epsilon} - z ) ( I + R_\epsilon ( z ) V ) \\
& = 
 ( I + V R_\epsilon ( z )\rho ) ( I + V R_\epsilon ( z ) ( 1 - \rho ) ) (  Q_{\epsilon} - z ) . \end{split}
\end{equation}
Noting that 
\[ ( I + V R_\epsilon ( z ) ( 1 - \rho ) )^{-1} = I - V R_\epsilon ( z ) ( 1 - \rho ) \]
we obtain from \eqref{eq:Repsz} and \eqref{eq:RPeps} that
\begin{gather}  
\label{eq:RVeps} 
\begin{gathered} R_{ V , \epsilon } ( z ) = R_\epsilon ( z ) ( I + V R_\epsilon ( z ) \rho )^{-1} 
( I - V R_\epsilon ( z ) ( 1 - \rho ) ) , \\ 
\delta < \arg z < 3 \pi/2 - \delta , \ \ | z| \gg 1 , 
\end{gathered}
\end{gather}
where for large $ |z| $, $  I + V R_\epsilon ( z ) \rho $ is invertible 
by a Neumann series argument. Since $ z \mapsto V R_\epsilon ( z ) \rho $ is 
a holomorphic family of compact operators for $ - \pi/4 < \arg z < 3 \pi /4 $
(see Lemma \ref{l:3} for the case $ \epsilon > 0 $), 
$ z \mapsto ( I + V R_\epsilon ( z ) \rho ) ^{-1} $
is a meromorphic family operators in the same range of $ z$. (For 
$ \epsilon >0 $ the meromorphy is in fact valid for $ z \in \CC $ -- 
see \cite[\S C.4]{res}.) The formula \eqref{eq:RVeps} remains valid
for that range of $ z$ with boundedness on $ L^2 $ for $ \epsilon > 0 $.
For $ \epsilon = 0 $ we note that
\[    ( I - V R_0 ( z ) ( 1 - \rho ) ) \ ,  (  I + V R_0 ( z ) )^{-1} 
 : L^2_{\rm{comp}} \to L^2_{\rm{comp}},  \ \ 
 R_0 ( z ) : L^2_{\rm{comp} } \to L^2_{loc} , \]
and we obtain the meromorphic continuation of $ R_{V, 0} ( z ) :
L^2_{\rm{comp} } \to L^2_{loc} $. Arguing as in the proof of
\cite[Theorem 3.23]{res} we obtain the multiplicity formula \eqref{eq:lmult}.
(This can also seen using complex scaling as reviewed in the proof of
Theorem \ref{t:2} below.)
\end{proof}

\section{Proof of convergence}
\label{poc}

The proof of convergence is based on Lemma \ref{l:mult} and on the 
following lemma in which we use the complex variable techniques of \S\S \ref{cs},\ref{dho}.

\begin{lem}
\label{l:fbound}
For $ \chi \in \CIc ( \RR^n ) $ consider 
\begin{equation} 
\label{eq:Teps}   T^\chi_\epsilon ( z )  := \chi ( - \Delta - i \epsilon x^2 - z )^{-1} x^2  ( - \Delta - z )^{-1} \chi , \ \  0 < \arg z < 3 \pi /2 . 
\end{equation} 
Then $ T_\epsilon^\chi $ continues to a holomorphic family of operators
\[  T_\epsilon^\chi ( z ) : L^2 \to L^2  , \ \ - \pi/4 < \arg z < 7 \pi /4. \]
If $ \Omega \Subset \{ z: - \pi /4 < \arg z < 3 \pi /2 \} $ then 
there exists $ C = C_{\Omega, \chi } $ (independent of $ \epsilon $) such that
\begin{equation}
\label{eq:lfbound}   \| T_\epsilon^\chi ( z ) \|_{ L^2 \to L^2 } \leq C , \ \ 
z \in \Omega, \ \ \epsilon > 0 . \end{equation}
\end{lem}
\begin{proof}
In the notation of \eqref{eq:Reps} we see that for $ \delta < \arg z < 3 \pi/2 - \delta $, $ |z | > \delta $,
\[  \chi ( R_{ \epsilon } ( z ) - R_0 ( z ) ) \chi = i \epsilon \chi R_{ \epsilon } ( z ) x^2 
R_0 ( z ) \chi , \]
where we note that, for in our range of $ z $, $ R_0 ( z ) \chi : L^2 \to 
e^{ - c_\delta |x| } L^2 $ (by looking, for instance at the explicit 
formulas for the resolvent, see \cite[\S 3.1]{res}, or by 
conjugation with exponential weights) and consequently $ x^2 R_0 ( z ) \chi : L^2 
\to L^2 $. Hence
\begin{equation}
\label{eq:RTeps}
T_\epsilon^\chi ( z ) = - \frac{ i } {\epsilon } (\chi R_\epsilon ( z ) \chi - 
\chi R_0 ( z ) \chi ) . \end{equation}
The right hand side is holomorphic for $ - \pi/4 < \arg z < 5 \pi/4 $ 
which provides holomorphic continuation of $ T_\epsilon^\chi ( z ) $, $ \epsilon > 0 $. 

We now use Lemmas \ref{l:2} and \ref{l:5}. For that we choose $r_0 $ in
the definition of $ \Gamma_\theta $ large enough so that $ \supp \chi \subset 
B ( 0 , r_0 ) $ and take $ \theta = \pi/8 $. Then we have
\begin{equation}  
\label{eq:Tepsz}
\begin{split} 
T_\epsilon^\chi ( z ) & = - \frac{ i } {\epsilon } \left( \chi ( Q_{\epsilon , \theta } -  z )^{-1}  \chi - \chi ( Q_{0, \theta } - z )^{-1} \chi \right) \\
& = 
 \chi ( Q_{\epsilon , \theta } -  z )^{-1} x_\theta^2 ( Q_{0, \theta } -z )^{-1} 
  \chi , 
  \end{split} \end{equation}
where, in the notation of \eqref{eq:defP}, $ x_\theta := x|_{\Gamma_\theta } $. 
We now note that for $ z \in \Omega $, 
\begin{equation}
\label{eq:Qexpd}   ( Q_{0, \theta } -z )^{-1} \chi : L^2 ( \Gamma_\theta ) \to 
e^{ - c_\Omega |x| } L^2 ( \Gamma_\theta ) .
\end{equation}
This can be seen by conjugation by exponential weights or by 
constructing a parametrix for $ Q_{0, \theta }$ as in the proof of 
Lemma \ref{l:5} and using the explicit properties of 
$ ( - e^{ - 2 i \theta } \Delta - z )^{-1} = e^{ 2i \theta } R_0 ( e^{2 i \theta} 
z ) $. From this and Lemma \ref{l:5} we obtain 
\[ \| ( Q_{\epsilon , \theta } -  z )^{-1} x_\theta^2 ( Q_{0, \theta } -z )^{-1} 
  \chi  \|_{ L^2 \to L^2 } \leq C_\Omega , \ \ z \in \Omega. \]
Inserting this into \eqref{eq:Tepsz} concludes the proof. 
\end{proof}

We can now state a stronger version of Theorem \ref{t:1} 
formulated using  the projections appearing in \eqref{eq:Pieps}:

\begin{thm}
\label{t:2}
Suppose that $ - \pi/4 < \arg z < 5\pi/4 $ and that $ m ( z) = m \geq 0 $, where
$ m ( z) $ is the multiplicity of the resonance of  $ 
 P := - \Delta + V $ at $ z $ -- see \eqref{eq:mz}. 

Then there exists $ \epsilon_0 $ and $ \delta$ such that
for $ 0 < \epsilon \leq \epsilon_0 $, $ P_\epsilon = - \Delta +V - i 
\epsilon x^2 $ has  
$ m $ eigenvalues in $ D ( z , \delta) $:
\begin{equation}
\label{eq:specP1}
\tr \Pi_\epsilon = m , \ \  \Pi_\epsilon := \frac 1 { 2 \pi i } \int_{\partial D ( z, \delta ) }
( \zeta - P_\epsilon )^{-1} d\zeta , \ \ \Pi_\epsilon^2 = \Pi_\epsilon,
 \end{equation}
and for any $ \chi \in \CIc ( \RR^n ) $, 
\begin{equation}
\label{eq:chiPchi} 
  \chi \Pi_\epsilon \chi \in C^\infty ( [ 0 , \epsilon_0 ] ,                         
  {\mathcal L } ( L^2 ( \RR^n ) , L^2 ( \RR^n ) ) ).  
\end{equation} 
\end{thm}

\medskip
\noindent
{\bf Remarks.} 1.
Notation $ f \in C^1 ( [ a , b ]) $ means that
$f$, and $f'$ continuous in $ [ a, b ] $;
here $f'(a),f'(b)$ are the left and right derivatives
of $ f $ at those points. By induction we then define $ C^k ( [ a, b ] ) $
and $ C^\infty ( [ a , b] ) $. In view of \eqref{eq:zje} we cannnot expect 
analytic dependence on $ \epsilon $.

\medskip
\noindent
2. 
For $ \chi \equiv 1 $ on $ \supp V $, $ m( z )  = \rank
\chi \Pi_0 \chi $ and \eqref{eq:chiPchi} shows the convergence 
of resonant states in the case of simple resonances. As the 
proof shows a stronger statement is obtained by using the complex
scaled operators: for $ \theta = \pi/8 $, 
\begin{gather}
\label{eq:chiPtheta}
\begin{gathered}
\Pi_{\epsilon, \theta } := 
\frac 1 { 2 \pi i } \int_{\partial D ( z, \delta ) }
( \zeta - P_{\epsilon, \theta })^{-1} d\zeta , 
\ \ P_{\epsilon, \theta } = -\Delta|_{\Gamma_\theta} - i 
\epsilon ( x|_{\Gamma_\theta} )^2 + V , 
\\
\Pi_{\epsilon, \theta } \in C ( [ 0 , \epsilon_0 ) ; 
\mathcal L^1 ( L^2 ( \Gamma_\theta ) , L^2 ( \Gamma_\theta ))), \ \ 
\  \Pi_{\epsilon, \theta } \chi \in \CI ( [ 0 , \epsilon_
0 ) ; 
\mathcal L^1 ( L^2 ( \Gamma_\theta ) , L^2 ( \Gamma_\theta )))  ,
\end{gathered}
\end{gather}
where $ \Gamma_\theta $ is the deformation defined in \eqref{eq:defGth}. 
\medskip

\begin{proof}
We first note that \eqref{eq:RVeps} and Lemma \ref{l:5}  imply that for $ 
-\pi/4 \leq - 2 \theta < \arg z < 2 \pi - 2 \theta $, $ \epsilon \geq 0 $, 
\begin{equation}
\label{eq:RVeps1}
(P_{\epsilon, \theta}  - z )^{-1} = 
( Q_{ \epsilon, \theta } - z )^{-1} ( I + V R_{\epsilon } ( z ) 
\rho )^{-1} ( I - V ( Q_{\epsilon , \theta } - z)^{-1} ( 1-  \rho ) ) .
\end{equation}
Since $ z \mapsto ( Q_{ \epsilon, \theta } - z )^{-1} $ is a holomorphic 
family in our range of $ z$'s, the Gohberg-Sigal theory -- see \cite[\S C.4]{res} -- 
shows that the poles of $ (P_{\epsilon, \theta}  - z )^{-1}  $ with 
$ \arg z > - 2 \theta $ are independent of $ 0 \leq \theta \leq \pi/8  $ 
and 
\[  \tr \frac 1 { 2 \pi i } \oint (P_{\epsilon, \theta}  - \zeta  )^{-1} d\zeta
= \tr \frac 1 { 2 \pi i } \oint (P_{\epsilon}  - \zeta  )^{-1} d\zeta, \ \ 
\epsilon > 0 . \]
If in the definition of $ \Gamma_\theta $ we take $ r_0 $ large enough
so that $ \supp \chi \subset B ( 0 , r_0 ) $ then Lemmas \ref{l:2} and \ref{l:5} 
show that 
$ \chi \Pi_{\epsilon, \theta} \chi = \chi \Pi_\epsilon \chi $.

Hence it is enough to prove \eqref{eq:chiPtheta}. For that we note that
in the notation of Lemma~\ref{l:fbound},
\[ \begin{split} ( I + V R_\epsilon ( z ) \rho )^{-1} - ( I + V R_0 ( z ) \rho)^{-1} 
& = i \epsilon  ( I + V R_\epsilon ( z ) \rho )^{-1}  T^\rho_\epsilon ( z ) 
 ( I + V R_0 ( z ) \rho )^{-1} \\
 & = \mathcal O_z ( \epsilon \| 
 ( I + V R_\epsilon ( z ) \rho )^{-1} \|_{ L^2 \to L^2} ) : L^2 
 \to L^2.
 \end{split} 
\]
We can now apply the Gohberg--Sigal--Rouch\'e theorem \cite[Theorem C.9]{res}
to see that the poles of $  ( I + V R_0 ( z ) \rho )^{-1} $ 
and $  ( I + V R_\epsilon ( z ) \rho )^{-1} $ coincide with multiplicities.
This and \eqref{eq:RVeps1} prove the first statement in \eqref{eq:chiPtheta}.
The second statement follows from differentiation and estimates similar to 
\eqref{eq:Qexpd}.
\end{proof}

\def\arXiv#1{\href{http://arxiv.org/abs/#1}{arXiv:#1}}

\end{document}